\documentclass{article}

\usepackage{ifthen} \newcommand{\vrs}{arxiv}

\usepackage[english]{babel}
\usepackage[a4paper,top=2cm,bottom=2cm,left=3cm,right=3cm,marginparwidth=1.75cm]{geometry}
\usepackage[T1]{fontenc}
\usepackage{xcolor,graphics}

\usepackage{xspace}
\usepackage{amsmath}
\usepackage{amssymb}
\usepackage{amsthm}
\usepackage{bm}
\usepackage{braket}
\usepackage{graphicx}
\usepackage{thm-restate}
\usepackage[colorlinks=false, hidelinks]{hyperref}
\usepackage[capitalize]{cleveref}

\newtheorem{theorem}{Theorem}[section]
\newtheorem{definition}[theorem]{Definition}
\newtheorem{lemma}[theorem]{Lemma}
\newtheorem{claim}[theorem]{Claim}
\newtheorem{corollary}[theorem]{Corollary}

\newtheorem*{remark}{Remark}

\newcommand{\Adj}{\mathrm{adj}}

\newcommand{\cI}{{\mathcal I}}

\newcommand{\N}{\mathbb{Z}_{\ge 0}}
\newcommand{\Q}{\mathbb{Q}}

\newcommand{\veca}{\ensuremath{\boldsymbol{a}}}
\newcommand{\vecb}{\ensuremath{\boldsymbol{b}}}
\newcommand{\vecc}{\ensuremath{\boldsymbol{c}}}

\newcommand{\vecu}{\ensuremath{\boldsymbol{u}}}

\newcommand{\cone}{\mathrm{cone}}
\newcommand{\intcone}{\mathrm{intcone}}
\newcommand{\latt}{{\mathcal L}}

\newcommand{\USS}{\textsc{Unbounded Subset Sum}\xspace}
\newcommand{\USSs}{\textsc{USS}\xspace}
\newcommand{\KN}{\textsc{Knapsack}\xspace}
\newcommand{\ILP}{\textsc{Integer Linear Programming}\xspace}
\newcommand{\ILPE}{\textsc{Integer Linear Programming with Equalities}\xspace}
\newcommand{\ILPs}{\textsc{ILP}\xspace}
\newcommand{\ILPEs}{\textsc{ILPE}\xspace}

\newcommand{\suppress}[1]{}
\newcommand{\MILP}{\textsc{Heterogeneous Integer Linear Programming}\xspace}
\newcommand{\MILPs}{\textsc{HILP}\xspace}

\newcommand{\vecv}{\ensuremath{\boldsymbol{v}}}

\newcommand{\tveca}{\ensuremath{\widetilde{\boldsymbol{a}}}}

\newcommand{\vecw}{\ensuremath{\boldsymbol{w}}}
\newcommand{\vecx}{\ensuremath{\boldsymbol{x}}}
\newcommand{\vecy}{\ensuremath{\boldsymbol{y}}}
\newcommand{\vecz}{\ensuremath{\boldsymbol{z}}}

\newcommand{\eps}{\varepsilon}

\DeclareMathOperator{\lat}{\mathcal{L}}
\DeclareMathOperator{\basis}{\mathbf{B}}

\DeclareMathOperator{\R}{\mathbb{R}}
\DeclareMathOperator{\Z}{\mathbb{Z}}

\newcommand{\floor}[1]{\lfloor{#1}\rfloor}
\newcommand{\ceil}[1]{\lceil{#1}\rceil}

\providecommand{\norm}[1]{\lVert#1\rVert}

\newcommand{\transpose}[1]{\ensuremath{#1^{\scriptscriptstyle T}}}

\newcommand{\poly}{\mathrm{poly}}

\renewcommand{\vec}[1]{\ensuremath{\boldsymbol{#1}}}

\ifthenelse{\equal{\vrs}{Debug}}{
\newcommand{\dnote}[1]{\textcolor{cyan}{[Divesh: #1]}}
\newcommand{\knote}[1]{\textcolor{teal}{[\textbf{Karol}: #1]}}
\newcommand{\old}[1]{{\leavevmode\color{blue}[#1]}}
\newcommand{\mnote}[1]{\textcolor{red}{[Miklos: #1]}}
\newcommand{\anote}[1]{\textcolor{orange}{[Antoine: #1]}}
}{
\newcommand{\dnote}[1]{}
\newcommand{\knote}[1]{}
\newcommand{\old}[1]{{}}
\newcommand{\mnote}[1]{}
\newcommand{\anote}[1]{}
}

\usepackage[framemethod=TikZ]{mdframed}

\title{Polynomial Time Algorithms for Integer Programming and Unbounded Subset Sum in the Total Regime}

\ifthenelse{\equal{\vrs}{submission}}{
    \author{Anonymous}
}{
\author{
    Divesh Aggarwal\footnote{Centre for Quantum Technologies and Department of Computer Science, NUS. This work was supported by the NRF investigatorship
grant, NRF-NRFI09-0005.}
    \and
    Antoine Joux\footnote{CISPA Helmholtz Center for Information Security, Germany. This work has been supported by the European Union's H2020 Programme under grant agreement number ERC-669891.}
    \and
    Miklos Santha\footnote{Centre for Quantum Technologies, National University of Singapore and CNRS, IRIF, Université de Paris. This research is supported by the National Research Foundation, Singapore and A*STAR under its CQT Bridging Grant.}
    \and
    Karol W\k{e}grzycki\footnote{Saarland University and Max Planck Institute for Informatics,
        Saarbr\"ucken, Germany. This work is part of the project TIPEA that has
    received funding from the European Research Council (ERC) under the European Union's Horizon
    2020 research and innovation programme (grant agreement No 850979).}
}}
\date{}

\begin{document}
\maketitle
\begin{abstract}

    \knote{blah blah blah}
The Unbounded Subset Sum (\USS) problem is an NP-hard computational problem
where the goal is to decide whether there exist non-negative integers $x_1,
\ldots, x_n$ such that $x_1 a_1 + \ldots + x_n a_n = b$, where $a_1 < \cdots <
a_n < b$ are distinct positive integers with $\gcd(a_1, \ldots, a_n)$ dividing
$b$. The problem can be solved in pseudopolynomial time, while specialized cases,
such as when $b$ exceeds the Frobenius number of $a_1, \ldots, a_n$ simplify to a
total problem where a solution always exists.

This paper explores the concept of totality in \USS. The challenge in this setting
is to actually find a solution, even though we know its existence is guaranteed.
We focus on the instances of \USS where solutions are guaranteed for large $b$.
We show that when $b$ is slightly greater than the Frobenius number, we can find
the solution to \USS in polynomial time.

We then show how our results extend to  \ILPE, highlighting conditions under
which \ILPE\ becomes total. We investigate the \emph{diagonal
Frobenius number}, which is the appropriate generalization of the Frobenius number to this context. In this setting, we give a  polynomial-time algorithm to find a solution of
\ILPE. The bound obtained from our algorithmic procedure for finding a solution almost matches the recent existential bound of Bach, Eisenbrand,
Rothvoss, and Weismantel (2024).

\end{abstract}
\ifthenelse{\equal{\vrs}{submission}}{
}{
\begin{picture}(0,0)
\put(450,-330)
{\hbox{\includegraphics[width=40px]{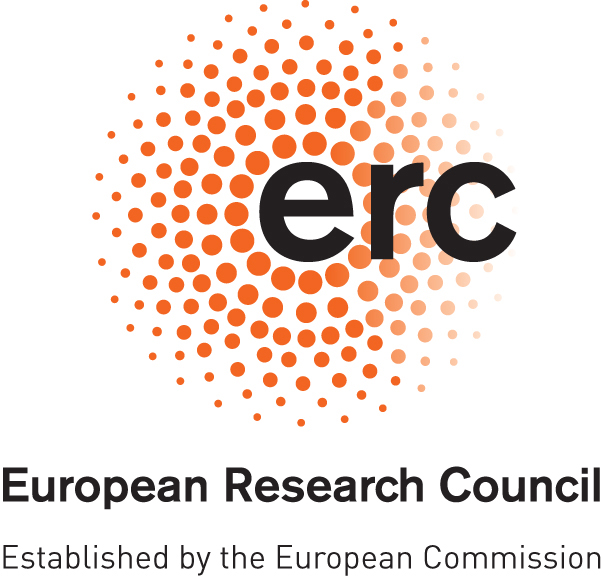}}}
\put(440,-400)
{\hbox{\includegraphics[width=60px]{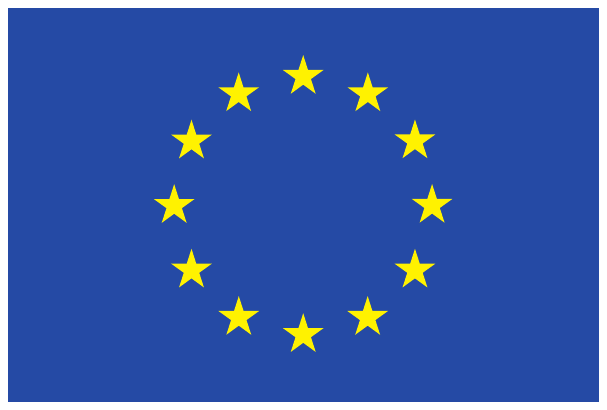}}}
\end{picture}
}

\section{Introduction}
In the \USS\ (in short notation \USSs) problem we are given $n$ distinct
positive integers $a_1 < \cdots < a_n$ with $\gcd(a_1, \ldots, a_n) = 1$ and a
target $b$. The task is to find non-negative integers $x_1,\ldots,x_n$ such that
$x_1 a_1 + \ldots + x_n a_n = b$, whenever such an $n$-tuple exists.  The
decision version of the problem is a variant of integer \KN\ and is well known
to be NP-complete, making \USSs\ NP-hard. It is a notoriously hard problem and
no algorithm which runs in time $2^{O(n)} (\log b)^{O(1)}$ is known for it. The
currently known fastest algorithm which solves \USSs, 
due to Reis and Rothvoss~\cite{DBLP:conf/focs/ReisR23}, and runs in time $(\log
n)^{O(n)}  (\log b)^{O(1)}$. This algorithm solves the general integer
programming problem of which \USSs\ is a special case.  Various
pseudopolynomial time algorithms using  dynamic programming were also given for
\USSs.  Supposing constant time arithmetic operations, Bringmann's
algorithm~\cite{DBLP:conf/soda/Bringmann17} works in time $O(b \log b)$, and the
one of Jansen and Rohwedder~\cite{DBLP:conf/innovations/JansenR19} in time
$O(a_n \log a_n \log(a_n +b))$.  For the setting of small $a_1$,  we have
algorithms in time $O(a_1^2 + n)$ by~\cite{HANSEN1996578}, in time $O(n a_1)$
by~\cite{bocker2007fast} and in time $O(a_1^{2-o(1)})$
by~\cite{DBLP:conf/soda/Klein22}.  These results are essentially tight assuming
widely believed fine-grained
hypotheses~\cite{DBLP:journals/jcss/AbboudBHS22,DBLP:conf/innovations/JansenR19,DBLP:conf/soda/Klein22}.
For constant $n$, the problem can be solved in polynomial time in the size of
the input numbers by generic integer programming
techniques~\cite{DBLP:journals/mor/Lenstra83, DBLP:conf/focs/ReisR23}.

Another case when the decision problem is easy to solve, and in fact trivializes, 
is when $b$ is sufficiently large. Indeed, for every $a_1, \ldots , a_n,$ there exists a largest integer
denoted by $g(a_1, \ldots , a_n)$ for which there is no solution, implying
that for every $b > g(a_1, \ldots , a_n)$, the problem always has a solution. This integer is called the
Frobenius number of $a_1, \ldots , a_n$, in honour of Frobenius~\cite{b2f296f6-cfc1-3d96-b687-b7c753783415} who first raised the problem of computing or estimating $g(a_1, \ldots , a_n)$.
There is a
substantial literature dedicated to this problem, including~\cite{heap1964graph, heap1965linear,hujter1987exact, kannan1992lattice, Alfonsin2005TheDF}. For $n=2$ the Frobenius number $g(a_1, a_2)=a_1a_2 - a_1 -a_2$
was determined by Sylvester~\cite{f85f7aa3-538b-330b-bb4c-261e026a757e} and relatively sharp estimates are known for $n=3$~\cite{Beck2003RefinedUB,
2009SbMat.200..597U}. In the general case the essentially sharpest upper bound 
of $g(a_1, \ldots , a_n) \leq 2 a_{n-1} \floor{\frac{a_n}{n}} - a_n$ is due to 
Erd\H{o}s and Graham~\cite{erdos1972linear} which is within a constant factor of the actual value. Variants and for some instances improvements of this bound can be found in ~\cite{Vitek_1976, Selmer1977, dixmier1990proof}. The exact computation of 
the Frobenius number is known to be NP-hard~\cite{ramirez1996complexity}.

\subparagraph*{Total Problems}

As we said, the decision version of \USSs\ is trivial when 
$b > g(a_1, \ldots , a_n)$ because the answer is always `yes'.
This makes the search problem total in the sense that there is always a solution.
It is worth to emphasize that the totality does not arise from a promise but rather from a mathematical property of the input. 
Furthermore, totality doesn't make the search problem necessarily easy to solve, and the
complexity of \USS in such instances is an
interesting research topic.  In fact, given any polynomila -time computable upper bound 
$u(a_1, \ldots, a_n) > g(a_1, \ldots , a_n),$ it is easy to verify 
that $b \geq u(a_1, \ldots, a_n).$ Thus, in this situation, the problem belongs to the complexity class TFNP.

Let us recall that, the class TFNP, introduced by Megiddo and Papadimitriou~\cite{MP91j}, consists of NP-search problems with total relations. It is known that no problem in TFNP can be NP-hard unless NP equals co-NP~\cite{Johnson1985HowEI, MP91j}. 
TFNP is believed not to have any complete problems.
Research on TFNP has mostly concentrated on sub-classes that can (also) be
defined by interesting complete problems. Examples include pure Nash Equilibrium
in a congestion game for PLS~\cite{Fabrikant2004TheCO}, Nash Equilibrium in a
two-player game~\cite{DBLP:journals/jacm/ChenDT09} or Multichromatic Simplex for
a Sperner coloring~\cite{DBLP:journals/tcs/ChenD09} for PPAD, and  Constrained Short Integer Solution for PPP~\cite{DBLP:conf/focs/SotirakiZZ18}.
However, for some important problems in TFNP, such as factoring or discrete
logarithm, which are not believed to be complete in any of these subclasses, the
research focus is on finding
the fastest possible algorithm. Our results are examples of this line of work.

\subparagraph*{Generalization of the Unbounded Subset Sum}
\ILPE\ (in short \ILPEs) is the natural generalization to higher dimensions of \USSs. In this problem, we are given 
$\vec{A}\in \Z^{d\times n}$ and $\vecb\in\mathbb Z^d$ and the
task is to find $\vecx \in \N^n$ such that $\vec{A}\vecx = \vecb$ if there
exists such an $\vecx$. In a variant, \ILP\ (in short \ILPs) on the same input
asks for $\vecx \in \N^n$ such that $\vec{A}\vecx \leq \vecb$. The two problems are easily inter-reducible in polynomial time.

Similarly to \USSs, the complexity of \ILPEs\ has been thoroughly studied. 
All known algorithms are exponential either in $n$ or in $\log L$, where
by definition $L$ is the  size of the maximum input number, that is
$L = \max \{\norm{\vec{A}}_\infty, \norm{\vecb}_\infty \}$.
The first algorithm of the former category was given by Lenstra~\cite{DBLP:journals/mor/Lenstra83} whose algorithm runs in time
$2^{O(n^3)} \poly(d \log L).$
This algorithm
provides a polynomial time solution when the number of variables is constant.
In a sequence of works~\cite{Kannan1987MinkowskisCB, Frank1987AnAO, Vempala2012IntegerPL} the polynomial for $n$ in the exponent was substantially
improved, and as it was mentioned for \USSs, the currently known fastest algorithm
due to Reis and Rothvoss~\cite{DBLP:conf/focs/ReisR23}, works in time
$2^{O(n \log \log n)} \poly(d \log L).$ The first algorithm of the second category,
which runs in time $n^{O(d)} (dL)^{O(d^2)},$
was presented by Papadimitriou~\cite{Papadimitriou81}. The currently fastest 
algorithm here, due to Jansen and
Rohwedder~\cite{DBLP:conf/innovations/JansenR19}, uses $O(d \norm{\vec A}_\infty )^{d}
+ O(dn)$ arithmetic operations.

Inspired by the existence of large targets which make \USSs\ a total problem
and by the relatively fast algorithms we could find in that case, we ask the
analogous question for the general \ILPEs. Our second result specifies conditions
on $\vecb$ under which \ILPEs\ becomes total and we are able to give a
polynomial time solution for it.

\subparagraph*{Total Regime of Integer Linear Programming}
The natural idea to generalize the Frobenius number to \ILPE would be to say that when an integer $\vec{b}$ (i) lies ``\emph{deep inside}'' a positive cone, and (ii) is in the lattice generated by
$\vec{A}$ then the answer to \ILPE is always positive. This,
however, is not true in general. Consider the following example:
\begin{displaymath}
    \vec{A} =
    \begin{pmatrix}
        9 & 10 & 9\\
        0 & 0  & 1
    \end{pmatrix}
    \in \mathbb{Z}^{2 \times 3}
    \text{ and }
    \vec{b} = 
    \begin{pmatrix}
        M\\
        M
    \end{pmatrix}
    \in \mathbb{Z}^2
    .
\end{displaymath}
Note that the lattice generated by $\vec{A}$ is $\mathbb{Z}^2$ and $\vec{b}$
lies ``\emph{deep-inside}'' the positive quarter for a large enough $M \in \N$. Nevertheless,
no matter how large the integer $M$ is, the target $\vec{b}$
cannot be obtained as a positive integral combination of columns of the matrix $\vec{A}$.
Hence, we need a more elaborate condition to properly generalize the Frobenius number. The
generalization we will consider, and which is arguably the appropriate one, is
the \emph{diagonal Frobenius number} $g(\vec{A})$ introduced by Aliev and
Henk~\cite{aliev2010feasibility}.

Consider the following three sets of points: $\cone(\vec{A}) = \{
\vec{A} \vecx \mid \vecx \in \mathbb{R}^n_{\ge 0}\}$, $\intcone(\vec{A}) =
\{\vec{A} \vecx \mid
\vecx \in \mathbb{Z}^n_{\ge 0}\},$ and $\latt(\vec{A})$ the lattice generated by
$\vec{A} \in \mathbb{Z}^{d \times n}$. The
diagonal Frobenius number $g(\vec{A})$ is the smallest non-negative integer $t$ such that
$$ \text{for every } \vecz \in \{\vec{A} \vecx \mid \vecx \ge  t 
\cdot \vec{1}\} \cap
\latt(\vec{A}) \text{ implies } \vecz \in \intcone(\vec{A}).$$
Recently, the existential statements of diagonal Frobenius numbers have been
used by Cslovjecsek et al.~\cite{DBLP:conf/soda/CslovjecsekKLPP24} to solve
two-stage stochastic programs and by Guttenberg et
al.~\cite{DBLP:conf/concur/GuttenbergRE23} to study geometric properties of
the Vector Addition Systems. The exact bounds are usually studied with respect
to the $\det(\vec{A}\transpose{\vec{A}})$ parameter. For example, Aliev and
Henk~\cite{aliev2010feasibility} show that the diagonal Frobenius number of
$\vec{A}$ is at most: 
\begin{displaymath}
    g(\vec{A}) \le \frac{(n-m)}{2}\sqrt{n \cdot \det(\vec{A} \transpose{\vec{A}})}.
\end{displaymath}

Very recently, Bach et
al.~\cite{bach2024forallexiststatementspseudopolynomialtime} considered the
diagonal Frobenius number parameterized by $\norm{\vec{A}}_\infty$. In that
setting, they improve the result of~\cite{aliev2010feasibility} and show that:
\begin{equation}\label{eq:fritz}
    g(\vec{A}) \le d \cdot \left(2d \cdot \norm{\vec{A}}_\infty+1\right)^d.
\end{equation}
Note that this bound is independent of $n$.

Several other generalizations of the Frobenius problem are already known in
the literature, for example, the $s$-Frobenius
number~\cite{fukshansky2011bounds}, semigroups~\cite{fel2006frobenius}, and
higher dimensions~\cite{FAN2015533} (see~\cite{shallit2008frobenius} for
references). Also, other restricted versions of ILP with polynomial-time
solutions have been considered in the literature. Perhaps the best-known
example is the case where the constraint matrix of the program is totally
unimodular. In this case, the linear programming relaxation is naturally
integral. Artmann, Weismantel, and
Zenklusen~\cite{DBLP:conf/stoc/ArtmannWZ17} extended this setting and gave a
polynomial-time algorithm for bimodular ILP, where all minors of the
constraint matrix are bounded in absolute value by 2. Many algorithms
obtained in the FPT context work in polynomial time when the respective
parameter is constant. Such examples include $n$-fold
ILP~\cite{DBLP:journals/disopt/LoeraHOW08, DBLP:journals/mp/HemmeckeOR13,
DBLP:journals/siamdm/JansenLR20}, three-fold
ILP~\cite{DBLP:conf/soda/ChenM18}, and more generally, ILP with block
structure~\cite{DBLP:conf/soda/CslovjecsekEHRW21,
DBLP:journals/mp/KnopKLMO23}.

\subsection{Our contribution} 

Our first result addresses the USS problem 
for a series of functions $u_t(a_1, \ldots, a_n)$, where the target $b$ is at least
$u_t(a_1, \ldots, a_n)$. These functions are at least as big as
the Erd\H{o}s and Graham bound of $\frac{a_n^2}{n-1}$, but less than $a_n^2$.

\begin{restatable}{theorem}{thmUSS}
    \label{thm:USS}
    Let $k$ be a non-negative integer. There is a $\poly(n,\log b)
\cdot (\log k)^{O(k)}$ time algorithm that given an {\rm \USS} instance $(n, a_1, \ldots, a_n, b)$ such that $b \ge \frac{a_i^2}{i-1}$, for
all $k < i \leq n$, and $\gcd(a_1, \ldots, a_n)$ divides $b$, finds $x_1, \ldots, x_n \in \N$ such that $\sum_{i=1}^n a_i x_i = b$.
\end{restatable}

This means that when $b$ is only greater than $u_n$ then the running time is the ILP bound of~\cite{DBLP:conf/focs/ReisR23} which we have anyhow without totality, but when
$b$ is greater than $u_0$, or as a matter of fact $u_c$, for some $c \in \N$,
then the running time is polynomial, see \cref{cor:USS}. Additionally, the complexity of our algorithm
smoothly transitions between these two extremes. The algorithm itself is
constructed by induction on the number of items, see~\cref{sec:USS}.

Inspired by the existence of large targets that make \USSs a total problem and by
the relatively fast algorithms we discovered in that case, we ask the analogous
question for \ILPEs. Our second result specifies conditions on $\vecb$
under which \ILPEs become total, allowing us to provide a polynomial-time
solution for them.

\begin{theorem}[Weaker version of Theorem~\ref{thm:pre-main}]\label{thm:simplified}
    Given target $\vec{b} \in \Z^d$ and constraint matrix $\vec{A} \in \Z^{d \times n}$ with column vectors $\veca_1, \ldots,
    \veca_n$ and  $\Delta = \max_{i=1}^n \ceil{\norm{\veca_i}}$.
    If 
    \begin{displaymath}
        t \ge (n-d) \cdot \Delta^d \text{ and } \vec{b} \in \lat(\vec{A}) \cap \{
        \vec{A} \cdot \vec{y} \mid \vec{y} \ge t \cdot \vec{1} \}
    \end{displaymath}
    then  we can find $\vec{x} \in \mathbb{Z}_{\ge 0}^n$ such that $\vec{b} =
    \vec{A} \cdot \vec{x}$ in polynomial time.
\end{theorem}
The proof idea for this result is as follows. We first use linear programming to find a real vector $\vec{y}  \ge t \cdot \vec{1}$ such that $\vec{b} = \vec{A} \cdot \vec{y}$. We then let $\vec{b} = \vec{v} + \vec{w}$, where $\vec{v} = \vec{A} \vec{z}$ for an integer vector $\vec{z}$, and $\vec{w} = \vec{A} \vec{u}$, where each coordinate of $\vec{u}$ is between $0$, and $1$. We know that $\vec{w} = \vec{b}- \vec{v}$ is in the lattice $\lat(\vec{A})$. Then we iteratively find the smallest non-negative integers $\beta_n, \ldots, \beta_{d+1}$ such that for any $i \in \{d+1, \ldots, n\}$, $\vec{w} - \sum_{j=i}^n \beta_i \vec{a}_i$ is in $\lat(\veca_1, \ldots, \veca_{i-1})$. Finally, there is a unique choice for $\beta_1, \ldots, \beta_d$. To conclude the proof, we show that $\beta_1, \ldots, \beta_d \ge -t$.

We note that~\cref{thm:simplified} implies that $g(\vec{A}) \le (n-d) \cdot \Delta^d$, hence it
offers an alternative bound on the diagonal
Frobenius number. Our result matches (up to the polynomial factors in $n$) the
currently best bound ~\eqref{eq:fritz}
of~\cite{bach2024forallexiststatementspseudopolynomialtime}. Moreover, our
result is algorithmic and allows us to find a solution in polynomial time.

We also give a closely matching lower bound for the diagonal Frobenius number.
In~\cref{thm:counter-example}, we show that our bound on $g(\vec{A})$ is tight
up to polynomial factors in $n$. 

\begin{theorem}[Theorem~\ref{thm:counter-example} simplified]\label{thm:simplified}
    For every $d \ge 2$ let $t = \frac{\Delta^d}{20 d}$. There exists 
    $\vec{b} \in \Z^d$ and constraint matrix $\vec{A} \in \Z^{d \times n}$ with column vectors $\veca_1, \ldots,
    \veca_n$ and  $\Delta = \max_{i=1}^n \ceil{\norm{\veca_i}}$ with the
    following property:
    \begin{displaymath}
         \vec{b} \in \lat(A) \cap \{
        \vec{A} \cdot \vec{x} \mid \vec{x} \ge t \cdot \vec{1} \}
    \end{displaymath}
    but there does not exist 
    $\vec{x} \in \mathbb{Z}_{\ge 0}$ such that $\vec{b} =
    \vec{A} \cdot \vec{x}$.
\end{theorem}

As a consequence,~\cref{thm:simplified} implies that $g(\vec{A}) \ge
\Omega(\Delta^d/d)$ and the
bound~\eqref{eq:fritz}
of~\cite{bach2024forallexiststatementspseudopolynomialtime} is tight (up to
polynomial factor in $d$). We note, that Aliev et al.~\cite{aliev2010feasibility}
also show a lower-bound on diagonal Frobenius number, but their result
is presented in terms of the $\det(\vec{A} \transpose{\vec{A}})$ parameter (see
also~\cite{aliev2011expected}). 

\section{Preliminaries}

\subparagraph*{Notation}
We denote by $\Z$ the set of integers and by $\N$ the set of non-negative integers. For a real number $\alpha$, we denote by $\lfloor \alpha \rfloor$, the largest integer less than or equal to $\alpha$, by $\lceil \alpha \rceil$, the smallest integer greater than or equal to $\alpha$, and by $\{\alpha\}$, the fractional part of $\alpha$, i.e., $\alpha - \lfloor \alpha \rfloor$. 

For $d \leq n$, a lattice $\lat \subset \R^d$ is the set of all integer linear combinations of $n$ vectors $\vec{A} = (\vec{a}_1,\ldots, \vec{a}_n) \in \R^{d \times n}$, i.e.,
\[
    \lat = \{ z_1 \vec{a}_1 + z_2 \vec{a}_2 + \cdots + z_n \vec{a}_n \ : \ z_i \in \Z\}
    \; .
\]
We also use the notation $\lat (\vec{a}_1,\ldots, \vec{a}_n)$ for $\lat.$
A basis of the lattice $\lat$ is $\vec{B} \in \R^{d \times n'}$ such that the column vectors of $\vec{B}$, $\vec{b}_1,\ldots, \vec{b}_{n'}$ are linearly independent and 

\[
    \lat = \lat(\vec{b}_1,\ldots, \vec{b}_{n'})
    \; .
\]
This is equivalent to saying that each of the vectors $\vec{a}_1, \ldots, \vec{a}_n$ can be written as integer combinations of $\vecb_1, \ldots, \vec{b}_{n'}$. 
An important geometric quantity associated with a lattice $\lat \subset \R^d$ is the determinant, $\det(\lat) := \det(\basis^T \basis)^{1/2}$. The determinant of the lattice is not dependent on the basis. If the basis is full rank, i.e., $n' = d$, then it is easy to see that $\det(\lat) = |\det(\basis)|$. Throughout the paper, $\|\vec{x}\| := (x_1^2 + \cdots + x_d^2)^{1/2}$ denotes the Euclidean norm of $\vecx \in \R^d.$

\subparagraph*{Definition of the Problems}
\begin{definition}[\ILP\ (\ILPs)]\label{def:ILP}
    On input $d, n \in \N$, $\vec{A}\in \Z^{d \times n}$, $\vec{b} \in \Z^{d}$, define the polytope $\mathcal{K}=\{\vec{x}\in \R^n: \vec{A}\vec{x} \leq \vec{b}\}$. The task is to find $\vec{x} \in \mathcal{K} \cap \N^n$ or output $\bot$ if $\mathcal{K} \cap \N^n = \emptyset$.
\end{definition}
If we replace the constraints by equality constraints, we get the following variant of the integer linear programming problem. 
\begin{definition}[\ILPE\ (\ILPEs)]\label{def:EqualityILP}
    On input $d, n \in \N$, $\vec{A}\in \Z^{d \times n}$, $\vec{b} \in \Z^{d}$, define the polytope $\mathcal{K}=\{\vec{x}\in \R^n: \vec{A}\vec{x} = \vec{b}\}$. The task is to find $\vec{x} \in \mathcal{K} \cap \N^n$ or output $\bot$ if $\mathcal{K} \cap \N^n = \emptyset$.
\end{definition}
Notice that the two variants of ILP mentioned above are computationally equivalent. 
\begin{itemize}
    \item An ILPE $\vec{A}\vec{x} = \vec{b}$ with $d$ equality constraints can be reduced to an ILP with $2d$ inequality constraints $\vec{A}\vec{x} \le \vec{b}$ and $-\vec{A}\vec{x} \le -\vec{b}$.
    \item An ILP with $d$ inequality constraints on $n$ variables $\vecx = (x_1, \ldots, x_n)$ given by $\vec{A}\vec{x} \le \vec{b}$ can be reduced an ILP with $d$ equality constraints on $n+d$ variables $(\vecx, \vecy) = (x_1, \ldots, x_n, y_1, \ldots, y_d)$ given by $\vec{A}\vec{x} + \vecy = \vec{b}$. 
\end{itemize}

The \USS problem is a special case of ILPE where $d=1$, and $\vec{A}$ has non-negative entries. 
\begin{definition}[\USS\ (\USSs)]
The (search version) of {\rm \USSs} is defined as follows: On input $n \in \N, a_1, \ldots, a_n, b \in \N$, find $x_1, \ldots, x_n \in \N$ such that $\sum_{i=1}^n a_ix_i = b$.
\end{definition}

Finally, we may have a combination of equality and inequality constraints to get the following variant. 

\begin{definition}[\MILP\ (\MILPs)]\label{def:MixedILP}
    On input $d_1, d_2, n \in \N$, $\vec{A}_1\in \Z^{d_1 \times n}$, $\vec{A}_2\in \Z^{d_2 \times n}$, $\vec{b}_1 \in \Z^{d_1}$ and $\vec{b}_2 \in \Z^{d_2}$, define the polytope $\mathcal{K}:=\{\vec{x}\in \R^n: \vec{A}_1\vec{x} \le  \vec{b}_1, \; \vec{A}_2 \vec{x} = \vec{b}_2\}$. The task is to find $\vec{x} \in \mathcal{K} \cap \N^n$ or output $\bot$ if $\mathcal{K} \cap \N^n = \emptyset$.
\end{definition}

\subparagraph*{Known Results}
We will use the following simplified version of a result due to Erd\H{o}s and Graham~\cite{erdos1972linear}, which was later improved by Dixmier~\cite{dixmier1990proof}.
\begin{theorem}
\label{thm:erdos_graham}
The {\rm \USSs} instance $(a_1, \ldots, a_n, b)$ has a solution if 
$gcd(a_1, \ldots, a_n)$ divides $b$, and 
$b \ge \frac{a_n^2}{n-1}$.
\end{theorem}

Additionally, we will need the following algorithm due to~\cite{DBLP:conf/focs/ReisR23} which is the state of the art algorithm for the ILP problem, and hence for all the variants mentioned above. 

\begin{theorem}
\label{thm:alg_ip}
There is an algorithm for the {\rm ILP} problem that runs in time $(\log
n)^{O(n)} \cdot \poly(|\cI|)$, where $n$ is the number of unknown variables, and $|\cI|$ is the total bitlength of the input instance $\vec{A}, \vec{b}, n, d$.
\end{theorem}

It is folklore that the LLL algorithm can also be run on the generating
families~\cite{DBLP:conf/stoc/Kannan83}. In particular, it is implemented
in~\cite{FPLLL}. For more recent discussion
see~\cite{DBLP:conf/eurocrypt/BennettGPS23}.

\begin{theorem}
\label{thm:findingbasis}
There is a polynomial time algorithm that given $\veca_1, \ldots, \veca_n \in \Z^d$ finds a basis of the lattice generated by all integer combinations of $\veca_1, \ldots, \veca_n$. 
\end{theorem}

\section{Algorithm for Unbounded Subset Sum}\label{sec:USS}

We will need the following technical lemma. 
\begin{lemma}
\label{lem:induction_hyptothesis}
Let $n, a_1 < \ldots < a_n,b$ be positive integers such that $b \ge \frac{a_n^2}{n-1}$, and $\gcd(a_1, \ldots, a_n)= d$. Then 
\[
d\left(b - a_n \cdot \left(b \cdot a_n^{-1} \pmod d \right)\right) > b \;.
\]
\end{lemma}
\begin{proof}
    Note that since $a_1, \ldots, a_{n-1}$ are distinct multiples of $d$, and $a_n > a_{n-1} > \cdots > a_1$, we have that $a_n/(n-1) > d$. Thus:
    \begin{align*}
      d\left(b - a_n \cdot \left(b \cdot a_n^{-1} \pmod d \right)\right) - b &= (d-1) b - a_n d \left(b \cdot a_n^{-1} \pmod d \right) \\
      &\ge (d-1)b - a_n d(d-1) 
      = (d-1)(b - a_n d) \\
      &\ge (d-1) \left(\frac{a_n^2}{n-1} - a_n d\right) >
       (d-1) \left(a_n d - a_n d\right) =0,
    \end{align*}
    as needed. 
\end{proof}

We now present the main result of this section, which is an algorithm that finds
a solution for the \USS problem when we have a stronger hypothesis on the input than the Erd\H{o}s-Graham condition which already guarantees the existence of a solution.
The stronger the condition, the faster the running time of the algorithm.

\thmUSS*
\begin{proof}
A solution always exists because of Theorem~\ref{thm:erdos_graham}. We assume without loss of generality that $\gcd(a_1, \ldots, a_n) = 1$ since if this is not the case, then it is equivalent to find a solution to an instance where we replace $a_1, \ldots, a_n, b$ by $\frac{a_1}{\gcd(a_1, \ldots, a_n)}, \ldots, \frac{a_n}{\gcd(a_1, \ldots, a_n)}, \frac{b}{\gcd(a_1, \ldots, a_n)}$ since $\frac{b}{\gcd(a_1, \ldots, a_n)} \ge \frac{a_i^2}{(i-1)\gcd(a_1, \ldots, a_n)^2}$ if $b \ge \frac{a_i^2}{i-1}$.

 Let $d = \gcd(a_1, \ldots, a_{n-1})$.

The algorithm ${\mathcal A}$ does the following. If $n\leq k$, the algorithm makes a call to the algorithm from Theorem~\ref{thm:alg_ip}. 
Otherwise, we proceed by setting $x_n = b \cdot a_n^{-1} \pmod d $ (i.e., the unique integer in $\{0,1, \ldots, d-1\}$ such that $d$ divides $b - a_n x_n$) and
\[
(x_1, \ldots, x_{n-1}) = {\mathcal A}\left(\frac{a_1}{d}, \ldots, \frac{a_{n-1}}{d}, \frac{b - x_n a_n}{d}\right).
\]

Notice that $d$ is relatively prime to $a_n$ since $\gcd(d,a_n)  = \gcd(a_1, \ldots, a_n)  = 1$ and hence $a_n^{-1} \pmod d$ exists. 
The correctness of the recursive step follows from
Lemma~\ref{lem:induction_hyptothesis}, which implies that for any $i$, such that
$k < i  \le n,$
\[\frac{b-x_n a_n}{d} > \frac{b}{d^2} \ge \frac{(a_i/d)^2}{i-1} \;. \]

Notice that in the case of $d=1$,  the recursive call  simply sets $x_n=0$ and proceeds to solve the problem $ {\mathcal A}(a_1, \ldots, a_{n-1}, b)$ with one variable fewer.
\end{proof}

The following corollary is an immediate of \cref{thm:USS}.

\begin{corollary}\label{cor:USS}
    Let $\eps > 0$ be a fixed constant. There is a polynomial time algorithm
    that given an instance $(a_1,\ldots,a_n,b)$ of {\rm \USS} such that $b \ge \eps
    \cdot a_n^2$ and $\gcd(a_1,\ldots,a_n)$ divides $b$, finds $x_1,\ldots,x_n
    \in \Z_{\ge 0}$
    such that $\sum_{i=1}^n x_i a_i = b$.
\end{corollary}

\section{Algorithm for Solving Variants of the ILP}

In this section, we focus on solving the \ILPE\ problem, which we restate here for the ease of the reader. Given $n$ distinct vectors $\veca_1 , \cdots , \veca_n \in \Z^d$ and a vector $\vecb \in \Z^d$, find $x_1, \ldots, x_n \in \N$ (if they exist) such that $\sum_{i=1}^n \veca_i x_i = \vecb$. We denote by $\vec{A} \in \Z^{d \times n}$ the matrix whose column vectors are $\vec{a}_1, \ldots, \vec{a}_n$. Thus, the problem can equivalently be stated as finding $\vec{x} \in \N^n$ such that $\vec{A} \vec{x} = \vecb$.

\begin{remark} 
We will restrict our attention to the assumption that $\vec{A}$ is a rank $d$ matrix. This is without loss of generality, since if $\vec{A}$ has rank $d' < d$, then we can use Gaussian elimination to find a subset of $d'$ rows that forms a matrix $\vec{A}' \in \Z^{d' \times n}$ (and let $\vecb'$ be $\vecb$ restricted to the same rows), and then it suffices to find a solution to $\vec{A}' \vec{x} = \vec{b}'$. This is because the rows of $\vec{A}$ are in the linear span of the rows of $\vec{A}'$, i.e., $\vec{A} = \vec{M} \vec{A}'$ for some $\vec{M} \in \R^{d \times d'}$. For a valid solution to exist, we must have that $\vec{M} \vec{b}' = \vec{b}$, which implies that $\vec{A} \vec{x} = \vec{b}$. 
\end{remark}

Additionally, we also assume that the first $d$ columns of $A$ are linearly independent. This can be easily achieved by appropriately permuting the columns.

\begin{definition}
    Given an integer  $V > 0$, we say that the matrix $\vec{A} \in \Z^{d \times n}$ satisfies the \emph{$V$-bounded} property if 
    \begin{itemize}
\item    for all  subsets $\{\vecv_1,\ldots,\vecv_{d-1}\}$ of $d-1$ vectors out of the first $d$ column vectors of $\vec{A}$, the determinant of $\lat(\vecv_1, \ldots, \vecv_{d-1})$ is at most $V$. 
\item The first $d$ columns of $\vec{A}$ are linearly independent.
    \end{itemize}
\end{definition}

Recall that $\Delta$ is the largest Euclidean norm among the vectors $\veca_1, \ldots, \veca_n$, rounded up to the next integer, i.e., $\Delta = \lceil \max_{i=1}^n \|\veca_i\| \rceil$.

\begin{theorem}
\label{thm:pre-main}
 There is an algorithm that takes as input a $V$-bounded matrix $\vec{A} \in
 \Z^{d \times n}$, and $\vec{b} \in \Z^d$, runs in time $\poly(n, \log{\Delta},
 \log \|\vec{b}\|)$, and does the following. Let $\vec{a}_1, \ldots, \vec{a}_n$ be the $n$ column vectors of $\vec{A}$. If there exist real $\alpha_1, \ldots, \alpha_d \ge (n-d) \cdot V \cdot \Delta$, and $\alpha_{d+1}, \ldots, \alpha_n \ge 0$ such that $\vecb = \sum_{i=1}^n \alpha_i \veca_i$, and $\vecb \in \lat(\veca_1, \ldots, \veca_n)$, then the algorithm finds $x_1, \ldots, x_n \in \N$ such that $\vecb = \sum_{i=1}^n \beta_i \veca_i$.
\end{theorem}
\begin{proof}
 Let $M = (n-d) \cdot V \cdot \Delta$. The algorithm begins by computing $V$ and $\Delta$, and then using linear programming to find some $\alpha_1,\ldots, \alpha_n \in \Q$ such that $\alpha_1, \ldots, \alpha_d \ge (n-d) \cdot V \cdot \Delta$, and $\alpha_{d+1}, \ldots, \alpha_n \ge 0$ such that $\vecb = \sum_{i=1}^n \alpha_i \veca_i$. Also, the algorithm finds a basis of $\lat(\vec{a}_1, \ldots, \vec{a}_n)$ using any efficient algorithm (for example,~\cite{DBLP:conf/eurocrypt/BennettGPS23}) and checks whether $\vecb \in \lat(\veca_1, \ldots, \veca_n)$. 
The algorithm then decomposes $\vecb$ as
$\vecb = \vecv + \vecw$, where
$\vecv = \sum_{i=1}^n \floor{\alpha_i } \veca_i $ and $\vecw = \sum_{i=1}^n \{\alpha_i \} \veca_i $. 
Since $\vecb \in \lat(\veca_1, \ldots, \veca_n)$, and $\vecv \in \lat(\veca_1, \ldots, \veca_n)$, we have that $\vecw \in \lat(\veca_1, \ldots, \veca_n)$. It is sufficient to find integers $\beta_1, \ldots, \beta_d \ge -M$, and $\beta_{d+1}, \ldots, \beta_n \ge 0$ and 
\[
\vecw = \sum_{i=1}^n \beta_i \veca_i\;.
\]
In the following, we show how to find such  integers $\beta_1, \ldots, \beta_n$. The desired solution is then given by $x_i = \lfloor \alpha_i \rfloor + \beta_i$ for $1 \le i \le n$. 

Let $V_i$ denote the volume of the fundamental parallelepiped of $\lat(\veca_1, \ldots, \veca_i)$.

\begin{claim}
\label{claim:invariant}
There exists a polynomial time algorithm that finds $\beta_n, \ldots, \beta_{d+1}$ such that for any $i \in \{d+1, \ldots, n\}$, we have that $0 \le \beta_i < V_{i-1}$, and $\vecw - \sum_{j=i+1}^n \beta_j \veca_j \in \latt(\veca_1, \ldots, \veca_i)$.
\end{claim}
\begin{proof}
We prove this claim by induction.
To see this, let $k \ge d+1$ be any positive integer. Suppose we have already found $\beta_n, \ldots, \beta_{k+1} \in \N$ that satisfy the above conditions for $i \ge k+1$. Let $\vecw' = \vecw - \sum_{j=k+1}^n \beta_j \veca_j \in \latt(\veca_1, \ldots, \veca_k)$. Our goal is to find $\beta_k \in \{0,1, \ldots, V_{k-1}-1\}$ such that $\vecw' - \beta_k \veca_k \in \latt(\veca_1, \ldots, \veca_{k-1})$.

Recall that the rank of the first $k-1$ vectors is $d$ because the first $d$ vectors are linearly independent. We first find a basis $\vec{B} \in \Z^{d \times d}$ of $\latt(\veca_1, \ldots, \veca_{k-1})$  using any efficient algorithm for computing a lattice basis given vectors generating the lattice, for example,~\cite{DBLP:conf/eurocrypt/BennettGPS23}. Note that $\vec{B}$ is an invertible matrix. Also, $|\text{det}(\vec{B})| = V_{k-1}$.

If $\vecw' \in \latt(\veca_1, \ldots, \veca_{k-1})$, i.e., $\vec{B}^{-1} \veca_k \in \Z^d$, then choosing $\beta_k = 0$ satisfies the desired condition.

Otherwise, we must have that for some $\gamma \in \Z$, $\vecw' - \gamma \veca_k \in \latt(\veca_1, \ldots, \veca_{k-1})$. This implies that $\vec{B}^{-1} (\vecw' - \gamma \veca_k) \in \Z^d$. 

We know that $\vec{B}^{-1} = \frac{\Adj(\vec{B})}{\det(\vec{B})}$, where  $\Adj(\vec{B})$, the adjugate of the matrix $\vec{B}$, is a $d \times d$ integer matrix, and $\det(\vec{B})$ is also an integer. This implies that all $d$ entries of
\[
\Adj(\vec{B})(\vecw' - \gamma \veca_k)
\]
are multiples of $\det(\vec{B})$. Moreover, $\vecw'$ and $\gamma \veca_k$ are integer vectors.
This leads to $d$ modular equations of the form
\[
x_i = \gamma y_i \pmod {\det(\vec{B})}  \;,
\]
where $x_i, y_i$ are the $i$-th coordinate of $\Adj(\vec{B}) \vecw'$ and $\Adj(\vec{B}) \veca_k$, respectively. 

Notice that there exists $\gamma \in \Z$ that satisfies these $d$ modular
equations (since  $\vecw' \in \lat(\veca_1, \ldots, \veca_k)$) and we want to
find any $\gamma$ that satisfies these $d$ modular equations. The existence of
such a $\gamma$ implies that $\gcd(y_i, \det(\vec{B}))$ divides $x_i$. Then, we can divide $x_i, y_i, \det(\vec{B})$ by $\gcd(y_i, \det(\vec{B}))$ to obtain
\[
    x_i' = \gamma y_i' \pmod {z_i'} \;,
\]
where $z_i'$ is a factor of $\det(\vec{B})$, and $y_i', z_i'$ are coprime. This
gives 
\[
\gamma = y_i'^{-1} x_i' \pmod {z_i'} \;.
\]
for $i \in [d]$. This gives $\gamma$ modulo $\text{lcm}(z_1', \ldots, z_d')$,
and any integer that is $\gamma$ modulo $\text{lcm}(z_1', \ldots, z_d')$ is such
that $\vecw' - \gamma \veca_k \in \latt(\veca_1, \ldots, \veca_{k-1})$. We can
find such a $\gamma \in \{0,1, \ldots, \text{lcm}(z_1', \ldots, z_d') -1 \}$
using a variant of the Chinese Remainder Theorem that does not require the
moduli to be coprime (see, for example~\cite{CRT}). Since $\text{lcm}(z_1',
\ldots, z_d')$ divides $V_{k-1}$ for any $\gamma$, we have that $\gamma \pmod
{V_{k-1}}$ is also $\gamma \pmod {\text{lcm}(z_1', \ldots, z_d')}$. 

Then let $\beta_k = \gamma$. We have that, $\vecw' - \beta_k \veca_k \in \latt(\veca_1, \ldots, \veca_{k-1})$.
\end{proof}
We now turn to find $\beta_1, \ldots, \beta_d$. Let 
\[
\vecw^* := \vecw - \sum_{j=d+1}^n \beta_j \veca_j = \sum_{i=1}^d\{\alpha_i\} \veca_i + \sum_{j=d+1}^n (\{\alpha_i\} - \beta_j) \veca_j \;.
\]
We know by Claim~\ref{claim:invariant} that $\vecw^* \in \latt(\veca_1, \ldots,
\veca_d)$. Notice, that since $\veca_1, \ldots, \veca_d$ are linearly independent, there is a unique linear combination of $\veca_1, \ldots, \veca_d$ that is equal to $\vecw^*$, i.e.,
\[\vecw^* = \sum_{i=1}^d \beta_i \veca_i\;,\]
which implies that 
\begin{equation}\label{eq:beta1tod}
\sum_{i=1}^d (\beta_i - \{\alpha_i\})\veca_i = \sum_{j=d+1}^n (\{\alpha_i\} - \beta_j) \veca_j \;.
\end{equation}
Also, $\beta_1, \ldots,\beta_d \in \Z$ since $\vecw^* \in \latt(\veca_1, \ldots, \veca_d)$. It is easy to compute $(\beta_1, \ldots, \beta_d)=(\veca_1, \ldots, \veca_d)^{-1}\vecw^*$.

To argue for correctness, it is enough to prove that $\beta_i \ge -M$. By symmetry, it is
sufficient to prove that $\beta_d \ge -M$. Let the projection of $\veca_d$ orthogonal to $\veca_1, \ldots, \veca_{d-1}$ be $\tveca_d$.  We project both sides in the
direction of $\tveca_d$. Let $\pi_{\tveca_d}(\vecu)$ be the projection of any vector $\vecu \in \R^d$ in the direction of $\tveca_d$. We get by projecting both sides of Equation~\ref{eq:beta1tod} in the direction of $\tveca_d$ that 
\begin{align*}
|\beta_d - \{\alpha_d\}| \|\tveca_d\| &= \|\pi_{\tveca_d}(\sum_{j=d+1}^n (\{\alpha_i\} - \beta_j) \veca_j)\|\\
&\le \| \sum_{j=d+1}^n (\{\alpha_j\} - \beta_j) \veca_j\| \\
&\le \sum_{j=d+1}^n V_{j-1} \Delta\;,
\end{align*}
using triangle inequality, and that $\beta_j \in \{0,1, \ldots, V_{j-1}-1\}$, and hence $\{\alpha_j\} - \beta_j \in [-V_{j-1}+1, 1)$.
Notice that for any $j \ge d+1$, $\latt(\veca_1, \ldots, \veca_{j-1})$ is a superlattice of $\latt(\veca_1, \ldots, \veca_d)$, and hence $V_{j-1} \le V_d$. 

Also, 
\[
\|\tveca_d\| = \frac{V_d}{V_{d-1}} \;.
\]
Thus,
\[
|\beta_d - \{\alpha_d\}| \le \frac{\sum_{j=d+1}^d V_{j-1} \Delta}{\|\tveca_d\|} \le \frac{(n-d)V_d}{V_d/V_{d-1}} = \Delta \cdot V_{d-1} \cdot (n-d)  \le \Delta \cdot V \cdot (n-d) \;,
\]
using that $V_{d-1} \le V$. This implies that $\beta_d \ge \{\alpha_d\} - (n-d) \Delta V \ge -(n-d) \Delta V$, since $\{\alpha_d\} \ge 0$.
\end{proof}

\begin{remark}
The algorithm in Theorem~\ref{thm:pre-main} relies on being given a $V$-bounded
matrix $\vec{A}$. In time $n^d \cdot \poly(n, \log{\Delta})$, we can try all possible combination of $d$ vectors to find $d$ linearly independent column vectors in $\vec{A}$, say $\veca_{i_1}, \ldots, \veca_{i_d}$ that for some value of $V$ satisfy the following two conditions.
\begin{itemize}
    \item The volume of the lattice generated by any $d-1$ of these $d$ vectors is at most $V$.
    \item There exist non-negative real numbers $\alpha_1, \ldots, \alpha_n$ with $\alpha_{i_1}, \ldots, \alpha_{i_d} \ge (n-d) \cdot V \cdot \Delta$  such that $\vecb = \sum_{i=1}^n \alpha_i \veca_i$.
\end{itemize}
If such indices $i_1, \ldots, i_d$ exist, then the algorithm finds a non-negative integer solution by permuting the columns so that the first $d$ columns are $\veca_{i_1}, \ldots, \veca_{i_d}$, and then running the algorithm from Theorem~\ref{thm:pre-main}.
\end{remark}

Given any $\vec{A}\in \Z^{d \times n}$ of rank $d$ with column vectors $\veca_1, \ldots, \veca_n$, it is easy to find via Gaussian elimination, $d$ linearly independent vectors $\veca_{i_1}, \ldots, \veca_{i_d}$, and the corresponding value of $V$ is at most $\left(\max_{i=1}^n \|\veca_i\|\right)^d$. Our main theorem is thus an immediate corollary of Theorem~\ref{thm:pre-main}
\begin{theorem}
    \label{thm:main}
There is an algorithm is given $\vec{A} \in \Z^{d
\times n}$ with column vectors $\veca_1, \ldots, \veca_n$, and $\vecb \in
\lat(\veca_1, \ldots, \veca_n)$ such that there exist $\alpha_1, \ldots,
\alpha_d \ge (n-d) \cdot \left(\max_{i=1}^n \|\veca_i\|\right)^d$, and
$\alpha_{d+1}, \ldots, \alpha_n \ge 0$ such that $\vecb = \sum_{i=1}^n \alpha_i
\veca_i$, and vectors $\veca_1, \ldots, \veca_d$ are linearly independent. The algorithm
finds $\beta_1, \ldots, \beta_n \in \N$ such that $\vecb = \sum_{i=1}^n \beta_i
\veca_i$, and runs in time $\poly(n, \log{\Delta}, \log{\|\vec{b}\|})$.
\end{theorem}

As stated earlier, a standard ILP with inequalities, or more generally, a heterogeneous ILP can be reduced to an ILP with equalities. As it turns out, Theorem~\ref{thm:main} gives a better bound for an ILP, or heterogeneous ILP, when solved via this reduction. 

\begin{corollary}
    \label{cor:MixedILP}
There is an algorithm that does the following.  Given $d_1, d_2, n \in \N$,
$\vec{A}_1\in \Z^{d_1 \times n}$, $\vec{A}_2\in \Z^{d_2 \times n}$ with column
vectors $\veca_1^{(2)}, \ldots, \veca_n^{(2)} \in \Z^{d_2}$, $\vec{b}_1 \in
\Z^{d_1}$ and $\vec{b}_2 \in \lat(\veca_1^{(2)}, \ldots, \veca_n^{(2)})$. Let
polytope $\mathcal{K}:=\{\vec{x}\in \R^n: \vec{A}_1\vec{x} \le  \vec{b}_1, \;
\vec{A}_2 \vec{x} = \vec{b}_2\}$, and let $\veca_1, \ldots, \veca_n \in \Z^{d_1
+ d_2}$ be vectors formed by concatenating the column vectors of $\vec{A}_1$ and
$\vec{A}_2$ such that the first $d_2$ vectors are linearly independent, and let
$\Delta = \max_{i=1}^n \|\veca_i\|$. The algorithm runs in time $\poly(n,
\log{\Delta}, \log \|\vecb\|)$ and finds $\vec{x} \in \mathcal{K} \cap \N^n$ assuming the following condition holds. There exist $\alpha_1, \ldots, \alpha_{d_2}, \gamma_1, \ldots, \gamma_{d_1} \ge (n-d_1 - d_2) \cdot \Delta^{d_2+1}$, and $\alpha_{d_2+1}, \ldots, \alpha_n \ge 0$ such that $\vec{A}_1 (\alpha_1, \ldots, \alpha_n)^T + (\gamma_1, \ldots, \gamma_{d_1})^T = \vecb_1$ and $\vec{A}_2 (\alpha_1, \ldots, \alpha_n)^T = \vecb_2$.
\end{corollary}
\begin{proof}
We do a standard reduction from heterogeneous ILP to ILP with equalities by introducing variables $\vecy = (y_1, \ldots, y_{d_1})^T$ such that the equations then become $\vecy + \vec{A}_1 \vecx = \vecb_1, \vec{A}_2 \vecx = \vecb_2$. These set of equations can be represented by the following matrix equation. 
\[
\begin{bmatrix}
\vec{I_{d_1}} & \vec{A}_1\\
\vec{0} & \vec{A}_2
\end{bmatrix} 
\begin{bmatrix}
\vec{y}\\
\vec{x}
\end{bmatrix} = 
\begin{bmatrix}
\vec{b}_1\\
\vec{b}_2 
\end{bmatrix}\;.
\]
We thereby obtain the result by observing that the volume of any $d_1 + d_2 -1$ out of the first $d_1+d_2$ columns is upper-bounded by $\max(\Delta^{d_2} \cdot 1^{d_1 - 1}, \Delta^{d_2-1} \cdot 1^{d_1}) = \Delta^{d_2}$. 
\end{proof}
Notice that if all the constraints are inequality constraints, then we get the following. 

\begin{corollary}
    \label{cor:ILP}
There is an algorithm that runs in time polynomial in the size of the input and
does the following.  The algorithm is given $d, n \in \N$, $\vec{A}\in \Z^{d \times n}$
such that $\veca_1, \ldots, \veca_n \in \Z^{d}$, $\vec{b} \in \Z^{d}$ are column
vectors of $\vec{A}$ and the first $d$ of them are linearly independent
Let $\Delta = \max_{i=1}^n \|\veca_i\|$ and 
let polytope $\mathcal{K}:=\{\vec{x}\in \R^n: \vec{A}\vec{x} \le  \vec{b}\}$. The algorithm finds $\vec{x} \in
\mathcal{K} \cap \N^n$ assuming the following conditions hold: There exist
$\gamma_1, \ldots, \gamma_{d} \ge (n-d) \cdot \Delta$, and $\alpha_1, \ldots, \alpha_n \ge 0$ such that $\vec{A} (\alpha_1, \ldots, \alpha_n)^T + (\gamma_1, \ldots, \gamma_d)^T = \vecb$.
\end{corollary}

\section{Lower Bound for the Totality Condition for ILP with Equalities}

In the previous section, in Theorem~\ref{thm:main}, we showed a sufficient condition under which an ILP always has a solution, i.e., it is a total problem. Furthermore, we gave a polynomial-time algorithm for finding such a solution. In this section, we show that the condition we obtained in Theorem~\ref{thm:main} is almost tight.

\begin{theorem} \label{thm:counter-example}
    For any $d \ge 2$, there exist $\vec{A} \in \Z^{d \times (d+1)}$ with column
    vectors $\veca_1, \ldots, \veca_{d+1}$, and $\vecb \in \lat(\veca_1, \ldots,
    \veca_{d+1})$ such that there exist $\alpha_1, \ldots, \alpha_{d+1} >
    \frac{\left(\max_{i=1}^{d+1} \|\veca_i\|\right)^d}{20\sqrt{d}}$ such that $\vecb
    = \sum_{i=1}^{d+1} \alpha_i \veca_i$ but there do not exist non-negative
    integers $\beta_1, \ldots, \beta_{d+1} \in \N$ such that $\vecb =
    \sum_{i=1}^{d+1} \beta_i \veca_i$. 
\end{theorem}

\begin{proof}
By~\cite{HB88}, there exist $c > 0$ such that there are at least $d$ distinct primes between $cd^2$, and $cd^2(1-1/d)$. Let $c \ge 5$. Let $p_1, \ldots, p_d$ be $d$ distinct primes between $d^*$, and $d^*(1-1/d)$.  Let $P = p_1 \cdots p_d$. Let $\Delta = \max_{i=1}^{d} p_i$. Then for all $i \in [d]$, $p_i \ge \Delta(1-1/d)$.  Also, let $p_{d+1}$ be a prime such that $\frac{\Delta}{2\sqrt{d}} \le p_{d+1} \le  \frac{\Delta}{\sqrt{d}}$. 
 
For $i= 1, \ldots, d$, let $\vec{a}_i = p_i \vec{e}_i \in \Z^d$, i.e., a vector
with $p_i$ in $i$-th coordinate and $0$, otherwise. Also, let $\veca_{d+1} =
(p_{d+1}, \ldots, p_{d+1}) \in \Z^d$. Note that for all $i \in [d+1]$,
$\|\vec{a}_i\| \le \Delta$. We first observe that $\lat(\veca_1, \ldots,
\veca_{d+1}) = \Z^d$. To see this, consider any vector $\vecc= (c_1, \ldots,
c_d) \in \Z^d$. This vector can be written as an integer combination of
$\veca_1, \ldots, \veca_{d+1}$ as follows. By the Chinese Remainder Theorem,
there is a unique integer $\gamma_{d+1}$ in $\{0,1, \ldots, P-1\}$ such that
$\gamma_{d+1} \equiv p_{d+1}^{-1} c_i \pmod {p_i}$ for $1 \le i \le d$. For $1
\le i \le d$, let $\gamma_i = \frac{c_i - \gamma_{d+1} p_{d+1}}{p_i}$. Note that by the
above reasoning $\gamma_1,\ldots,\gamma_d$ are integers. Moreover, it holds that $\vecc = \sum_{i=1}^{d+1} \gamma_i \veca_i$. 

Now, let $\vecb = (p_{d+1}(P-1) - p_1, p_{d+1}(P-1) - p_2, \ldots, p_{d+1}(P-1)
- p_d)^T \in \Z^d$. We first show that there exist large real $\alpha_1,\ldots,
\alpha_{d+1}$ such that $\vecb = \sum_{i=1}^{d+1} \alpha_i \vecb$. Let
$\alpha_{d+1} = \frac{P}{2}$, and for $1 \le i \le d$, $\alpha_i = P \cdot \frac{p_{d+1}}{2p_i} - 1- \frac{p_{d+1}}{p_i}$. Then $\vecb = \sum_{i=1}^{d+1}
\alpha_i \veca_i$, and for all $i \in [d]$, \[
\alpha_i \ge
\frac{\Delta}{2\sqrt{d}} \cdot \frac{\Delta^{d-1}(1-1/d)^{d-1}}{2} - 2 \ge 
\frac{\Delta^d}{4e \sqrt{d}} - 2 \ge \frac{\Delta^d}{20 \sqrt{d}}\;,\] where we use the fact that for $d \ge 2$, $(1-1/d)^{d-1} \ge \frac{1}{e}$, and that $\Delta \ge 5d^2 (1-1/d) \ge 10$.

We now show that there does not exist non-negative integers $\beta_1, \ldots, \beta_{d+1}$ such that $\vecb = \sum_{i=1}^{d+1} \beta_i \veca_i$. Suppose there exist such non-negative integers. Then, we must have that for $i \in [d]$
\[
\beta_i p_i + \beta_{d+1} p_{d+1} = P p_{d+1} - p_i - p_{d+1}\;,
\]
which implies that $\beta_{d+1} \equiv -1 \pmod {p_i}$ for all $i \in [d]$. By Chinese Remainder Theorem, this implies $\beta_{d+1} \equiv -1 \pmod P$. Since $\beta_{d+1}$ is non-negative, we must have that $\beta_{d+1} \ge P-1$. This implies that for $1 \le i \le d$, $\beta_i p_i \le P p_{d+1} - p_i - p_{d+1} - (P-1)p_{d+1} <0$, which is a contradiction.
\end{proof}

\bibliographystyle{alpha}
\bibliography{References}

\newcommand{\etalchar}[1]{$^{#1}$}
\begin{thebibliography}{LHOW08}

\bibitem[ABHS22]{DBLP:journals/jcss/AbboudBHS22}
Amir Abboud, Karl Bringmann, Danny Hermelin, and Dvir Shabtay.
\newblock Scheduling lower bounds via {AND} subset sum.
\newblock {\em J. Comput. Syst. Sci.}, 127:29--40, 2022.

\bibitem[AH10]{aliev2010feasibility}
Iskander Aliev and Martin Henk.
\newblock Feasibility of integer knapsacks.
\newblock {\em SIAM Journal on Optimization}, 20(6):2978--2993, 2010.

\bibitem[AHH11]{aliev2011expected}
Iskander Aliev, Martin Henk, and Aicke Hinrichs.
\newblock {Expected Frobenius numbers}.
\newblock {\em Journal of Combinatorial Theory, Series A}, 118(2):525--531,
  2011.

\bibitem[AJ05]{Alfonsin2005TheDF}
Ram{\'i}rez Alfonsin and L.~Jorge.
\newblock {The Diophantine Frobenius problem}.
\newblock 2005.

\bibitem[AWZ17]{DBLP:conf/stoc/ArtmannWZ17}
Stephan Artmann, Robert Weismantel, and Rico Zenklusen.
\newblock A strongly polynomial algorithm for bimodular integer linear
  programming.
\newblock In {\em Proceedings of the 49th Annual {ACM} {SIGACT} Symposium on
  Theory of Computing, {STOC} 2017}, pages 1206--1219. {ACM}, 2017.

\bibitem[BERW24]{bach2024forallexiststatementspseudopolynomialtime}
Eleonore Bach, Friedrich Eisenbrand, Thomas Rothvoss, and Robert Weismantel.
\newblock Forall-exist statements in pseudopolynomial time, 2024.

\bibitem[BGPS23]{DBLP:conf/eurocrypt/BennettGPS23}
Huck Bennett, Atul Ganju, Pura Peetathawatchai, and Noah Stephens{-}Davidowitz.
\newblock {Just How Hard Are Rotations of \(\mathbb{Z}^n\) ? Algorithms and
  Cryptography with the Simplest Lattice}.
\newblock In {\em Advances in Cryptology - {EUROCRYPT} 2023 - 42nd Annual
  International Conference on the Theory and Applications of Cryptographic
  Techniques, Proceedings, Part {V}}, volume 14008 of {\em Lecture Notes in
  Computer Science}, pages 252--281. Springer, 2023.

\bibitem[BL07]{bocker2007fast}
Sebastian Bocker and Zsuzsanna Lipt{\'a}k.
\newblock A fast and simple algorithm for the money changing problem.
\newblock {\em Algorithmica}, 48(4):413--432, 2007.

\bibitem[Bra42]{b2f296f6-cfc1-3d96-b687-b7c753783415}
Alfred Brauer.
\newblock On a problem of partitions.
\newblock {\em American Journal of Mathematics}, 64(1):299--312, 1942.

\bibitem[Bri17]{DBLP:conf/soda/Bringmann17}
Karl Bringmann.
\newblock A near-linear pseudopolynomial time algorithm for subset sum.
\newblock In {\em Proceedings of the Twenty-Eighth Annual {ACM-SIAM} Symposium
  on Discrete Algorithms, {SODA} 2017}, pages 1073--1084. {SIAM}, 2017.

\bibitem[BZ03]{Beck2003RefinedUB}
Matthias Beck and Shelemyahu Zacks.
\newblock {Refined upper bounds for the linear Diophantine problem of
  Frobenius}.
\newblock {\em Adv. Appl. Math.}, 32:454--467, 2003.

\bibitem[CD09]{DBLP:journals/tcs/ChenD09}
Xi~Chen and Xiaotie Deng.
\newblock On the complexity of {2D} discrete fixed point problem.
\newblock {\em Theor. Comput. Sci.}, 410(44):4448--4456, 2009.

\bibitem[CDT09]{DBLP:journals/jacm/ChenDT09}
Xi~Chen, Xiaotie Deng, and Shang{-}Hua Teng.
\newblock Settling the complexity of computing two-player nash equilibria.
\newblock {\em J. {ACM}}, 56(3):14:1--14:57, 2009.

\bibitem[CEH{\etalchar{+}}21]{DBLP:conf/soda/CslovjecsekEHRW21}
Jana Cslovjecsek, Friedrich Eisenbrand, Christoph Hunkenschr{\"{o}}der, Lars
  Rohwedder, and Robert Weismantel.
\newblock Block-structured integer and linear programming in strongly
  polynomial and near linear time.
\newblock In D{\'{a}}niel Marx, editor, {\em Proceedings of the 2021 {ACM-SIAM}
  Symposium on Discrete Algorithms, {SODA} 2021, Virtual Conference, January 10
  - 13, 2021}, pages 1666--1681. {SIAM}, 2021.

\bibitem[CKL{\etalchar{+}}24]{DBLP:conf/soda/CslovjecsekKLPP24}
Jana Cslovjecsek, Martin Kouteck{\'{y}}, Alexandra Lassota, Micha\l{}
  Pilipczuk, and Adam Polak.
\newblock Parameterized algorithms for block-structured integer programs with
  large entries.
\newblock In {\em Proceedings of the 2024 {ACM-SIAM} Symposium on Discrete
  Algorithms, {SODA} 2024}, pages 740--751. {SIAM}, 2024.

\bibitem[CM18]{DBLP:conf/soda/ChenM18}
Lin Chen and D{\'{a}}niel Marx.
\newblock Covering a tree with rooted subtrees - parameterized and
  approximation algorithms.
\newblock In {\em Proceedings of the Twenty-Ninth Annual {ACM-SIAM} Symposium
  on Discrete Algorithms, {SODA} 2018}, pages 2801--2820. {SIAM}, 2018.

\bibitem[CRT]{CRT}
{Chinese Remainder Theorem}.
\newblock
  \url{https://cp-algorithms.com/algebra/chinese-remainder-theorem.html}.
\newblock [Online; accessed 24-June-2024].

\bibitem[Dix90]{dixmier1990proof}
Jacques Dixmier.
\newblock Proof of a conjecture by {E}rd{\H{o}}s and {G}raham concerning the
  problem of {F}robenius.
\newblock {\em Journal of Number Theory}, 34(2):198--209, 1990.

\bibitem[EG72]{erdos1972linear}
Paul Erd{\H{o}}s and Ronald Graham.
\newblock On a linear {D}iophantine problem of {F}robenius.
\newblock {\em Acta Arithmetica}, 1(21):399--408, 1972.

\bibitem[Fel06]{fel2006frobenius}
Leonid~G. Fel.
\newblock Frobenius problem for semigroups.
\newblock {\em Functional Analysis and Other Mathematics}, 1(2):119--157, 2006.

\bibitem[FPL24]{FPLLL}
{The FPLLL Development Team. FPLLL, a lattice reduction library, Version:
  5.4.2.}
\newblock Available at \url{https://github.com/fplll/fplll}, 2024.

\bibitem[FPT04]{Fabrikant2004TheCO}
Alex Fabrikant, Christos~H. Papadimitriou, and Kunal Talwar.
\newblock The complexity of pure {N}ash equilibria.
\newblock In {\em Symposium on the Theory of Computing}, 2004.

\bibitem[FS11]{fukshansky2011bounds}
Lenny Fukshansky and Achill Sch{\"u}rmann.
\newblock Bounds on generalized {F}robenius numbers.
\newblock {\em European Journal of Combinatorics}, 32(3):361--368, 2011.

\bibitem[FT87]{Frank1987AnAO}
Andr{\'a}s Frank and {\'E}va Tardos.
\newblock An application of simultaneous diophantine approximation in
  combinatorial optimization.
\newblock {\em Combinatorica}, 7:49--65, 1987.

\bibitem[GRE23]{DBLP:conf/concur/GuttenbergRE23}
Roland Guttenberg, Mikhail~A. Raskin, and Javier Esparza.
\newblock {Geometry of Reachability Sets of Vector Addition Systems}.
\newblock In {\em 34th International Conference on Concurrency Theory, {CONCUR}
  2023,}, volume 279 of {\em LIPIcs}, pages 6:1--6:16. Schloss Dagstuhl -
  Leibniz-Zentrum f{\"{u}}r Informatik, 2023.

\bibitem[HB88]{HB88}
D.R. Heath-Brown.
\newblock The number of primes in a short interval.
\newblock 1988.

\bibitem[hFRZ15]{FAN2015533}
Ai~hua Fan, Hui Rao, and Yuan Zhang.
\newblock {Higher dimensional Frobenius problem: Maximal saturated cone, growth
  function and rigidity}.
\newblock {\em Journal de Mathématiques Pures et Appliquées},
  104(3):533--560, 2015.

\bibitem[HL64]{heap1964graph}
B.R. Heap and M.S. Lynn.
\newblock A graph-theoretic algorithm for the solution of a linear diophantine
  problem of {F}robenius.
\newblock {\em Numerische Mathematik}, 6:346--354, 1964.

\bibitem[HL65]{heap1965linear}
B.R. Heap and M.S. Lynn.
\newblock On a linear diophantine problem of {F}robenius: an improved
  algorithm.
\newblock {\em Numerische Mathematik}, 7:226--231, 1965.

\bibitem[HOR13]{DBLP:journals/mp/HemmeckeOR13}
Raymond Hemmecke, Shmuel Onn, and Lyubov Romanchuk.
\newblock {N-fold integer programming in cubic time}.
\newblock {\em Math. Program.}, 137(1-2):325--341, 2013.

\bibitem[HR96]{HANSEN1996578}
Paul Hansen and Jennifer Ryan.
\newblock Testing integer knapsacks for feasibility.
\newblock {\em European Journal of Operational Research}, 88(3):578--582, 1996.

\bibitem[HV87]{hujter1987exact}
Mih{\'a}ly Hujter and B{\'e}la Vizv{\'a}ri.
\newblock The exact solutions to the {F}robenius problem with three variables.
\newblock {\em Journal of the Ramanujan Mathematical Society}, pages 117--143,
  1987.

\bibitem[JLR20]{DBLP:journals/siamdm/JansenLR20}
Klaus Jansen, Alexandra Lassota, and Lars Rohwedder.
\newblock {Near-Linear Time Algorithm for {$n$}-Fold ILPs via Color Coding}.
\newblock {\em {SIAM} J. Discret. Math.}, 34(4):2282--2299, 2020.

\bibitem[JPY85]{Johnson1985HowEI}
David~S. Johnson, Christos~H. Papadimitriou, and Mihalis Yannakakis.
\newblock How easy is local search?
\newblock {\em 26th Annual Symposium on Foundations of Computer Science (sfcs
  1985)}, pages 39--42, 1985.

\bibitem[JR19]{DBLP:conf/innovations/JansenR19}
Klaus Jansen and Lars Rohwedder.
\newblock {On Integer Programming and Convolution}.
\newblock In {\em 10th Innovations in Theoretical Computer Science Conference,
  {ITCS} 2019}, volume 124 of {\em LIPIcs}, pages 43:1--43:17. Schloss Dagstuhl
  - Leibniz-Zentrum f{\"{u}}r Informatik, 2019.

\bibitem[Kan83]{DBLP:conf/stoc/Kannan83}
Ravi Kannan.
\newblock {Improved Algorithms for Integer Programming and Related Lattice
  Problems}.
\newblock In {\em Proceedings of the 15th Annual {ACM} Symposium on Theory of
  Computing, 25-27 April, 1983, Boston, Massachusetts, {USA}}, pages 193--206.
  {ACM}, 1983.

\bibitem[Kan87]{Kannan1987MinkowskisCB}
Ravi Kannan.
\newblock Minkowski's convex body theorem and integer programming.
\newblock {\em Math. Oper. Res.}, 12:415--440, 1987.

\bibitem[Kan92]{kannan1992lattice}
Ravi Kannan.
\newblock Lattice translates of a polytope and the {F}robenius problem.
\newblock {\em Combinatorica}, 12(2):161--177, 1992.

\bibitem[KKL{\etalchar{+}}23]{DBLP:journals/mp/KnopKLMO23}
Dusan Knop, Martin Kouteck{\'{y}}, Asaf Levin, Matthias Mnich, and Shmuel Onn.
\newblock High-multiplicity {N}-fold {IP} via configuration {LP}.
\newblock {\em Math. Program.}, 200(1):199--227, 2023.

\bibitem[Kle22]{DBLP:conf/soda/Klein22}
Kim{-}Manuel Klein.
\newblock {On the Fine-Grained Complexity of the Unbounded SubsetSum and the
  Frobenius Problem}.
\newblock In {\em Proceedings of the 2022 {ACM-SIAM} Symposium on Discrete
  Algorithms, {SODA} 2022}, pages 3567--3582. {SIAM}, 2022.

\bibitem[Len83]{DBLP:journals/mor/Lenstra83}
Hendrik~W. Lenstra.
\newblock Integer programming with a fixed number of variables.
\newblock {\em Math. Oper. Res.}, 8(4):538--548, 1983.

\bibitem[LHOW08]{DBLP:journals/disopt/LoeraHOW08}
Jes{\'{u}}s A.~De Loera, Raymond Hemmecke, Shmuel Onn, and Robert Weismantel.
\newblock {N}-fold integer programming.
\newblock {\em Discret. Optim.}, 5(2):231--241, 2008.

\bibitem[MP91]{MP91j}
N.~Megiddo and C.~H. Papadimitriou.
\newblock On total functions, existence theorems and computational complexity.
\newblock {\em Theoretical Computer Science}, 81(2):317--324, 1991.

\bibitem[Pap81]{Papadimitriou81}
Christos~H. Papadimitriou.
\newblock On the complexity of integer programming.
\newblock {\em J. {ACM}}, 28(4):765--768, 1981.

\bibitem[RA96]{ramirez1996complexity}
Jorge~L Ram{\'\i}rez-Alfons{\'\i}n.
\newblock Complexity of the {F}robenius problem.
\newblock {\em Combinatorica}, 16:143--147, 1996.

\bibitem[RR23]{DBLP:conf/focs/ReisR23}
Victor Reis and Thomas Rothvoss.
\newblock The subspace flatness conjecture and faster integer programming.
\newblock In {\em 64th {IEEE} Annual Symposium on Foundations of Computer
  Science, {FOCS} 2023, Santa Cruz, CA, USA, November 6-9, 2023}, pages
  974--988. {IEEE}, 2023.

\bibitem[Sel77]{Selmer1977}
Ernst~S. Selmer.
\newblock On the linear diophantine problem of {F}robenius.
\newblock {\em Journal für die reine und angewandte Mathematik},
  0293\_0294:1--17, 1977.

\bibitem[Sha08]{shallit2008frobenius}
Jeffrey Shallit.
\newblock {The Frobenius problem and its generalizations}.
\newblock In {\em International Conference on Developments in Language Theory},
  pages 72--83. Springer, 2008.

\bibitem[Syl82]{f85f7aa3-538b-330b-bb4c-261e026a757e}
J.~J. Sylvester.
\newblock On subvariants, i.e. semi-invariants to binary quantics of an
  unlimited order.
\newblock {\em American Journal of Mathematics}, 5(1):79--136, 1882.

\bibitem[SZZ18]{DBLP:conf/focs/SotirakiZZ18}
Katerina Sotiraki, Manolis Zampetakis, and Giorgos Zirdelis.
\newblock {PPP-Completeness with Connections to Cryptography}.
\newblock In {\em 59th {IEEE} Annual Symposium on Foundations of Computer
  Science, {FOCS} 2018}, pages 148--158. {IEEE} Computer Society, 2018.

\bibitem[{Ust}09]{2009SbMat.200..597U}
Alexey~V. {Ustinov}.
\newblock {The solution of Arnold's problem on the weak asymptotics of
  {F}robenius numbers with three arguments}.
\newblock {\em Sbornik: Mathematics}, 200(4):597--627, April 2009.

\bibitem[VD12]{Vempala2012IntegerPL}
Santosh~S. Vempala and Daniel Dadush.
\newblock Integer programming, lattice algorithms, and deterministic volume
  estimation.
\newblock 2012.

\bibitem[Vit76]{Vitek_1976}
Yehoshua Vitek.
\newblock {Bounds for a Linear Diophantine Problem of Frobenius, II}.
\newblock {\em Canadian Journal of Mathematics}, 28(6):1280–1288, 1976.

\end{thebibliography}
\end{document}